\documentclass[a4paper, 12pt, numbers=noenddot]{scrartcl}

\usepackage[T1]{fontenc}
\usepackage[utf8]{inputenc}
\usepackage[english]{babel}
\usepackage{xcolor}

\usepackage[left=1.8cm, right=1.8cm, bottom=2.7cm, top=2.7cm]{geometry}
\usepackage[onehalfspacing]{setspace}
\usepackage{ragged2e}
\usepackage{enumitem}
\usepackage{float}
\usepackage[section]{placeins}
\usepackage{verbatim}
\usepackage[headsepline, footsepline]{scrlayer-scrpage}
\usepackage{footnote}
\usepackage[hang,flushmargin]{footmisc}
\usepackage{anyfontsize}

\usepackage{amsmath}
\usepackage{amsthm}
\usepackage{newtxtext,newtxmath}
\newtheorem{proposition}{Proposition}[section]
\theoremstyle{remark}
\newtheorem{remark}{Remark}[section]
\theoremstyle{plain}

\usepackage{array}
\usepackage{multirow}
\usepackage{longtable}
\usepackage{tabularx}
\usepackage{makecell}
\usepackage{siunitx}
\usepackage{booktabs}
\usepackage{comment}

\definecolor{arcaBlue}{HTML}{000000}
\definecolor{arcaRed}{HTML}{000000}
\definecolor{arcaWhite}{HTML}{FFFFFF}
\definecolor{arcaBlack}{HTML}{000000}

\usepackage{graphicx}
\usepackage[most]{tcolorbox}
\usepackage{tikz}
\usepackage{pgfplots}
\pgfplotsset{compat=1.18}

\usepackage{csquotes}
\usepackage[backend=biber, style=numeric, sorting=none]{biblatex}

\DefineBibliographyStrings{italian}{%
    bibliography={Bibliografia}%
}
\usepackage{xurl}
\usepackage[colorlinks=true, linkcolor=black, citecolor=black, urlcolor=black]{hyperref}
\newcommand{\Title}{Reverse Stress Testing for Multivariate Scenarios}
\newcommand{\Subtitle}{A Conditional Framework for Stressed Time Series}

\newcommand{\ReportDate}{March 2025}

\newcommand{\Keywords}{Reverse stress testing; Scenario simulation; Multivariate stress scenarios; Empirical likelihood; Semiparametric methods; Nonparametric resampling; Market risk.}
\newcommand{\JELCodes}{C14; C15; C53; C58; G17; G32.}

\newcommand{\TitlePageAuthors}{%
    \par\noindent
    \makebox[\textwidth][c]{%
        \begin{tabular}{@{}c@{\hspace{2.5cm}}c@{}}
            {\fontsize{13}{16}\selectfont Michele Sparviero} &
            {\fontsize{13}{16}\selectfont Lorenzo Viola} \\[0.35em]
            {\small ARCA Fondi SGR} &
            {\small ARCA Fondi SGR} \\[0.35em]
            {\small\texttt{michele.sparviero@arcafondi.it}} &
            {\small\texttt{lorenzo.viola@arcafondi.it}}
        \end{tabular}%
    }%
    \par
}

\setkomafont{disposition}{\normalfont\bfseries}

\setkomafont{section}{\normalfont\large\bfseries}

\setkomafont{subsection}{\normalfont\normalsize\bfseries}

\setkomafont{subsubsection}{\normalfont\normalsize\itshape}

\setkomafont{paragraph}{\normalfont\normalsize\itshape}
\setkomafont{subparagraph}{\normalfont\normalsize\itshape}

\RedeclareSectionCommand[
  beforeskip=2.5ex plus 1ex minus .2ex,
  afterskip=1.2ex plus .2ex
]{section}

\RedeclareSectionCommand[
  beforeskip=2ex plus .8ex minus .2ex,
  afterskip=0.8ex plus .2ex
]{subsection}

\RedeclareSectionCommand[
  beforeskip=1.5ex plus .6ex minus .2ex,
  afterskip=0.6ex plus .2ex
]{subsubsection}

\clearpairofpagestyles

\setkomafont{pageheadfoot}{\normalfont\small}
\setkomafont{pagenumber}{\normalfont\small}

\newtcolorbox{boxK}{
    colback=white,
    colframe=black,
    boxrule=0.4pt,
    arc=0pt,
    left=6pt,
    right=6pt,
    top=6pt,
    bottom=6pt,
    before skip=1em,
    after skip=1em,
    fontupper=\small
}

\begin{document}

\hypersetup{pageanchor=false}
% ============================================================
%  Frontespizio in stile accademico (working paper)
%  Titolo, autori e abstract su un'unica pagina.
%  Sfondo bianco, testo interamente nero, nessuna grafica.
% ============================================================
\begin{titlepage}
\thispagestyle{empty}
\begin{singlespace}
\centering

\vspace*{1.3cm}

% --- Titolo ---
{\fontsize{20}{25}\selectfont\bfseries \Title\par}

\vspace{0.6cm}

% --- Sottotitolo ---
{\fontsize{14}{18}\selectfont \Subtitle\par}

\vspace{1.3cm}

% --- Autori, affiliazione, contatti ---
\TitlePageAuthors

\vspace{0.85cm}

% --- Data / versione ---
{\normalsize \ReportDate\par}

\vspace{1.2cm}

% --- Abstract, parole chiave e classificazione JEL ---
\begin{minipage}{0.82\textwidth}
    \begin{center}
        {\large\bfseries Abstract}
    \end{center}
    \vspace{0.35em}

    {\normalsize\justifying\noindent
        This paper develops a methodological framework for reverse stress
testing (RST) in which a multivariate stress scenario, coherent with the empirical dependence structure of a market, is reconstructed from a single exogenous shock prescribed on one asset class. The problem is formulated as the maximisation of the conditional density given the imposed shock, and is solved under three progressively weaker distributional assumptions. In the parametric setting, joint Gaussianity of the returns yields a closed-form modal scenario coinciding with the conditional mean of the non-shocked components. In the semiparametric setting, the modal
scenario is estimated nonparametrically through the empirical
likelihood methodology and the surrounding stressed
trajectories are generated via a Gaussian or Student-\emph{t} local sampling scheme. In the fully nonparametric setting, stressed trajectories are obtained by inverse-distance resampling of the historical observations within a Mahalanobis neighbourhood of the estimated scenario. The three variants are validated on real market data. The simulated
scenarios prove to be economically coherent and capable of
reproducing the standard risk--reward asymmetry observed in stressed
market regimes.
\par}

    \vspace{1.0em}

    {\small\raggedright\noindent
        \par}

    \vspace{0.45em}

    {\small\raggedright\noindent
        \textbf{JEL classification:} \JELCodes\par}
\end{minipage}

\vfill
\end{singlespace}
\end{titlepage}

\hypersetup{pageanchor=true}
\pagenumbering{arabic}
\setcounter{page}{1}
\pagestyle{scrheadings}

\section{Introduction}
In financial applications, scenario simulation plays a central role in assessing capital allocation, pricing, hedging strategies, and broader strategic decisions under plausible yet unobserved future market conditions. Among the most widely used methodologies, Monte Carlo simulation relies on the parametric random generation of price trajectories and represents a fundamental tool for derivative valuation and risk measurement. Bootstrapping methods, by contrast, use nonparametric resampling techniques to project historical time series, thereby providing a flexible framework for simulating market dynamics without imposing a fully specified distributional structure.

Beyond the simulation of expected or baseline dynamics, financial institutions commonly introduce deviations from central scenarios in order to assess portfolio resilience under adverse conditions. This is typically achieved through \textit{stress testing}, where one or more risk factors are shocked and the resulting impact on aggregate quantities---such as liabilities, portfolio value, solvency ratios, or maximum drawdown---is evaluated. Standard examples include the assessment of the portfolio loss generated by a $300$ basis point increase in interest rates or by a $30\%$ decline in equity markets.

A key limitation of conventional stress testing is that shocks are often imposed on individual asset classes in isolation. This approach may fail to capture the propagation of stress across the broader financial system. For instance, a $30\%$ decline in US equities would be unlikely to occur without affecting government bond yields, credit spreads, exchange rates, volatility, liquidity conditions, and other risk factors. Ignoring these interactions may therefore lead to internally inconsistent scenarios and to an incomplete assessment of portfolio vulnerability.

The methodology proposed here addresses this limitation by constructing a multivariate stress scenario that is consistent with the historical dependence structure among the relevant risk factors. The approach can be described in two stages:
\begin{enumerate}
    \item \textit{Estimation of the multivariate stress scenario.}\\
    Starting from an adverse outcome of interest---such as a target portfolio loss, a maximum drawdown threshold, or a critical solvency condition---the method identifies the combination of market shocks that is most consistent with both the specified outcome and the historical relationships among the underlying risk factors.

    \item \textit{Application of the stress scenario.}\\
    The estimated multivariate shock vector is then applied to the historical series, portfolio exposures, valuation models, or liability structure under analysis. This allows the impact of a coherent and system-wide stress event to be assessed across the quantities of interest, including portfolio valuation, risk measures, liquidity needs, or balance-sheet indicators.
\end{enumerate}

The main contribution of this framework lies in the preliminary generation of a coherent multivariate scenario before the portfolio impact is evaluated. In a traditional stress test, the direction of the analysis is typically
\[
  \text{Scenario} \;\longrightarrow\; \text{Loss}.
\]
The scenario is specified exogenously, and the resulting loss is computed. By contrast, in a \textbf{Reverse Stress Testing} (RST) framework, the analysis starts from an adverse outcome and works backwards to identify the market configurations that could plausibly generate it:
\[
  \text{Loss} \;\longrightarrow\; \text{Multivariate Scenario}.
\]
This inversion of perspective is the defining feature of Reverse Stress Testing. Rather than asking what the loss would be under a pre-defined shock, the method asks which joint movements in risk factors could lead to a given level of loss or financial distress. As a result, the framework provides a more informative and internally consistent view of portfolio vulnerabilities, especially when risks arise from the interaction of multiple market factors rather than from a single isolated shock.

\section{Reverse Stress Testing}

Consider a market composed of \(d\) asset classes, and denote by
\[
\hat r_{i,t}\in\mathbb{R},
\qquad i\in\{1,\dots,d\},\quad t\in\{1,\dots,T\},
\]
the daily return of the \(i\)-th asset class at time \(t\), observed
over a historical sample of length \(T\). The vector of returns at
time \(t\) is collected in
\(\hat r_{t}=(\hat r_{1,t},\dots,\hat r_{d,t})^{\top}\in\mathbb{R}^{d}\),
and the corresponding random vector of daily returns is denoted by
\(r=(r_{1},\dots,r_{d})^{\top}\), with joint density \(f_{r}\).

Without loss of generality, we assume that the asset class subject to
the exogenous shock is the first one. Accordingly, we partition the
return vector as \(r=(r_{1},\,r_{-1})\), where
\(r_{-1}=(r_{2},\dots,r_{d})^{\top}\in\mathbb{R}^{d-1}\) collects the
returns of the non-shocked asset classes. Throughout this work, the
subscript \(-1\) will denote the sub-vector (or sub-matrix, when
applied to a covariance matrix) obtained by removing the entry
(resp.\ row and column) associated with the shocked component.

The RST procedure is parametrised by four inputs:
\begin{itemize}
    \item the index \(j\in\{1,\dots,d\}\) of the asset class subject to
    the shock; without loss of generality, \(j=1\);
    \item the cumulative shock \(S\in\mathbb{R}\) prescribed on the
    selected asset class over the entire stress window;
    \item the starting date \(T_{\text{start}}\in\{1,\dots,T\}\) of the
    stress window;
    \item the length \(T_{\text{win}}\in\mathbb{N}\) of the stress
    window, measured in trading days.
\end{itemize}
The stress window is therefore the discrete time interval
\(\{T_{\text{start}},\dots,T_{\text{start}}+T_{\text{win}}\}\), within
which the procedure replaces the original returns of the shocked
asset class with realisations consistent with the imposed shock. Since the cumulative shock \(S\) refers to the entire stress window,
while the procedure operates on daily returns, it is convenient to
introduce the \emph{daily-equivalent shock}
\[
s \;:=\; (1+S)^{1/T_{\text{win}}}-1.
\]
As discussed in
Section~\ref{ch:applications}, alternative conversions based on
longer blocks (e.g.\ weekly or monthly) are also admissible, provided
the temporal base is kept consistent across the entire analysis.

The conditional density of the non-shocked returns \(r_{-1}\) given the
shocked return \(r_{1}\) is defined as
\[
f_{r_{-1}\mid r_{1}}(x_{-1}\mid x_{1})
\;=\;
\frac{f_{r}(x_{1},x_{-1})}{f_{r_{1}}(x_{1})},
\qquad x=(x_{1},x_{-1})\in\mathbb{R}^{d},
\]
and describes the joint law of the non-shocked components conditional
on a prescribed value of the shocked one. Within the RST framework,
the conditioning event of interest is the imposed shock
\(\{r_{1}=s\}\), and the conditional density
\(f_{r_{-1}\mid r_{1}}(\,\cdot\mid s)\) collects all joint
configurations of the remaining asset classes consistent with such a
shock.

Among the uncountably many configurations compatible with the
conditioning, RST singles out the one that maximises the conditional
density, namely
\begin{equation}\label{opt}
\bar r^{\,s}_{-1}
\;=\;
\arg\max_{x_{-1}\in\mathbb{R}^{d-1}}\;
f_{r_{-1}\mid r_{1}}(x_{-1}\mid s),
\qquad
\bar r^{\,s}\;=\;\bigl(s,\;\bar r^{\,s}_{-1}\bigr).
\end{equation}

The remainder of the paper is devoted to the solution of
problem~\ref{opt} under progressively weaker distributional
assumptions, paired in each case with a coherent scheme for the
simulation of stressed trajectories.
Section~\ref{para} solves problem~\ref{opt} in closed form under the
assumption that the return vector \(r\) is jointly Gaussian,
recovering the standard conditional-distribution formulae and the
associated sampling procedure.
Section~\ref{rst} relaxes the Gaussian assumption and develops a
semiparametric framework in which the modal scenario
\(\bar r^{\,s}\) is estimated nonparametrically via the
\emph{empirical likelihood} methodology of Owen~\cite{owenn}, while
the local sampling around \(\bar r^{\,s}\) is performed under a
Gaussian or Student-\emph{t} model.
Section~\ref{nonpar-sampling} discusses a fully nonparametric
sampling scheme, based on an inverse-distance reweighting of the
historical observations within a Mahalanobis neighbourhood of
\(\bar r^{\,s}\).
The methodology is finally validated on real market data in
Section~\ref{ch:applications}, where the three variants are compared
on a common experimental setup.
\section{Parametric Approach}\label{para}

\noindent The treatment of problem~\ref{opt} simplifies considerably once
an explicit parametric form is assumed for the joint density \(f_r\).
In what follows, we adopt the working hypothesis that the return vector
\(r\in\mathbb{R}^{d}\) is jointly Gaussian, so that
\[
f_r(x)
=
\frac{1}{(2\pi)^{d/2}\,\lvert\Sigma\rvert^{1/2}}
\exp\!\left\{
-\frac{1}{2}(x-\mu)^\top\Sigma^{-1}(x-\mu)
\right\},
\]
where \(\mu\in\mathbb{R}^{d}\) is the mean vector and
\(\Sigma\in\mathbb{R}^{d\times d}\) a symmetric, positive-definite
covariance matrix. Standard estimators are employed: \(\mu\) is
estimated by the sample mean and \(\Sigma\) by the sample covariance
matrix of historical returns. We refer to these estimates as
\emph{global} when they are computed on the full historical sample of
the market under consideration, and as \emph{local} when they are
restricted to observations recorded under specific market conditions
(for instance, periods of distress). The detailed construction of both
estimators is deferred to Appendix~B.

In order to solve problem~\ref{opt}, an explicit expression for the
conditional density \(f_{r_{-1}\mid r_1}\) is required. We first fix
the notation that will be used throughout. Without loss of generality,
we assume that the exogenous shock acts on the first component of the
return vector, denoted by \(r_1\), and we collect the remaining
components in the \((d-1)\)-dimensional vector
\[
r_{-1}
\;:=\;
(r_2,\dots,r_d)^{\top}\in\mathbb{R}^{d-1},
\]
so that \(r=(r_1,\,r_{-1}^{\top})^{\top}\). Accordingly, the mean vector and the covariance
matrix admit the conformable block decomposition
\[
\mu
=
\begin{pmatrix}
\mu_1\\
\mu_{-1}
\end{pmatrix},
\qquad
\Sigma
=
\begin{pmatrix}
\Sigma_{11} & \Sigma_{1,-1}\\[4pt]
\Sigma_{-1,1} & \Sigma_{-1,-1}
\end{pmatrix},
\]
where \(\mu_1\in\mathbb{R}\) and \(\Sigma_{11}\in\mathbb{R}_{>0}\) are
the (scalar) global sample mean and variance of \(r_1\),
\(\mu_{-1}\in\mathbb{R}^{d-1}\) is the vector of global sample means of
the non-shocked components, and the off-diagonal blocks satisfy
\(\Sigma_{1,-1}=\Sigma_{-1,1}^{\top}\) by symmetry of~\(\Sigma\). This
block representation is well suited to the present problem: the first
component is imposed exogenously through the stress scenario, while the
remaining components must be determined coherently with the estimated
dependence structure.

Under the multivariate Gaussian assumption, the conditional
distribution of \(r_{-1}\) given the event \(\{r_1=s\}\) is itself
Gaussian, namely
\[
r_{-1}\mid (r_1=s)
\sim
\mathcal{N}_{d-1}\!\left(
\mu_{-1\mid 1},\;\Sigma_{-1\mid 1}
\right),
\]
with conditional mean and covariance
\[
\mu_{-1\mid 1}
=
\mu_{-1}
+
\Sigma_{-1,1}\Sigma_{11}^{-1}(s-\mu_1),
\qquad
\Sigma_{-1\mid 1}
=
\Sigma_{-1,-1}
-
\Sigma_{-1,1}\Sigma_{11}^{-1}\Sigma_{1,-1}.
\]
We note that \(\Sigma_{-1\mid 1}\) is the Schur complement of the block
\(\Sigma_{11}\) in \(\Sigma\). Since \(\Sigma\) is positive definite,
this complement is itself positive definite, so the conditional density
is non-degenerate.

Since the conditional density \(f_{r_{-1}\mid r_1}\) is Gaussian, it is
unimodal and attains its maximum at the conditional mean
\(\mu_{-1\mid 1}\). The solution of problem~\ref{opt} on the
non-shocked components is therefore \(\mu_{-1\mid 1}\), and the full
stressed configuration is recovered by reinserting the fixed shocked
coordinate \(r_1=s\), yielding
\[
\bar r^{\,s}
=
\begin{pmatrix}
s\\
\mu_{-1\mid 1}
\end{pmatrix}.
\]
The vector \(\bar r^{\,s}\) admits a natural interpretation as the
most likely joint configuration of the returns, conditional on the
imposed shock \(r_1=s\) and on the estimated Gaussian dependence
structure.

The realisations \(\hat r^{\,s}\) are coherent both with the imposed
shock and with the estimated dependence among the risk factors;
repeated sampling traces out an approximation of the conditional law
of \(r_{-1}\mid r_1=s\) around the modal scenario \(\bar r^{\,s}\).

In practical applications, the covariance blocks entering the
conditional moments may be estimated either globally or locally,
depending on the analytical objective.

We finally emphasise that, within this framework, the stress scenario
is imposed by conditioning on the event \(\{r_1=s\}\): the shocked
component is not generated stochastically in the simulation but is
fixed by construction, whereas the remaining components are drawn from
the conditional law derived above.
\section{Semiparametric Approach}\label{rst}

Problem~\ref{opt} can be addressed without imposing any explicit
assumption on the form of the density \(f_{r}\). Direct estimation of
\(f_{r}\) is, however, far from straightforward: on the one hand, a
parametric approach risks imposing overly restrictive constraints and
underestimating the weight of the tails; on the other, a multivariate
nonparametric estimator is subject to the \emph{curse of
dimensionality}, with numerical instability and interpolations that
fail to reflect the true variability of the data.

To overcome these limitations, we adopt the \emph{empirical likelihood}
methodology of Owen~\cite{owenn}, which allows the likelihood function to be
constructed directly from the observed data, without any parametric
specification of \(f_{r}\). We shall show that, under suitable
conditions on the optimisation domain discussed in Appendix~A, the
solution \(\bar r^{\,s}\) admits a closed form and can be interpreted
as the empirical mean of a properly selected sub-sample.

Once \(\bar r^{\,s}\) has been determined, perturbed historical series
are generated by sampling: the original returns within the stress
window are replaced by realisations drawn around \(\bar r^{\,s}\), so
that the stress event is embedded in the resulting trajectories. For
the sampling step under stress we propose two parametric alternatives:
one based on the multivariate Gaussian density (Section~\ref{gauss})
and one on the multivariate Student-\emph{t} density
(Section~\ref{tstud}), which offer explicit control, respectively, of
the correlation structure and of the local tail thickness. The
resulting framework is therefore \emph{semiparametric}: the estimation
of the centre \(\bar r^{\,s}\) is nonparametric, while the local
generation of returns relies on a parametric model calibrated in a
neighbourhood of \(\bar r^{\,s}\). For completeness, a fully
nonparametric sampling scheme is also discussed in
Section~\ref{nonpar-sampling}.

Consider a market composed of \(d\) asset classes, the first of which
is subject to an exogenous shock of magnitude \(s\in\mathbb{R}\).
Denote by
\[
\hat r_{t}
\;=\;
(\hat r_{1,t},\dots,\hat r_{d,t})^{\top}\in\mathbb{R}^{d},
\qquad t=1,\dots,T,
\]
the realisation of the daily return vector at time \(t\), and by \(T\)
the size of the available historical sample. As in the previous
section, the subscript \(-1\) denotes the \(d-1\) returns associated
with the asset classes that are not subject to the shock.

In order to study the behaviour of the market in proximity of the
stress event, we restrict the historical sample to those dates at
which the return of the first asset class is consistent with the target
level \(s\). More precisely, given a buffer \(\varepsilon>0\), we
define the set of conditioned time indices
\[
\mathcal{I}^{s}
\;:=\;
\bigl\{\,t\in\{1,\dots,T\}\;:\;\lvert \hat r_{1,t}-s\rvert\le \varepsilon\,\bigr\},
\qquad
T^{s}\;:=\;\lvert\mathcal{I}^{s}\rvert,
\]
with \(T^{s}<T\). Without loss of generality, we re-index the elements
of \(\mathcal{I}^{s}\) as \(t=1,\dots,T^{s}\), so that the conditioned
sub-sample is denoted by
\[
\bigl\{\hat r_{t}\bigr\}_{t=1}^{T^{s}}\subset\mathbb{R}^{d},
\]
where, by construction, \(\hat r_{1,t}\approx s\) for every
\(t\in\{1,\dots,T^{s}\}\). This preprocessing step provides a practical device for recasting a conditional-probability optimisation problem as an unconditional one.

Consistently with the nonparametric approach, we assume that the
solution of problem~\ref{opt} belongs to the convex hull of the
conditioned observations, that is, that there exists a weight vector
\(w=(w_{1},\dots,w_{T^{s}})\) in the simplex
\[
\Delta_{T^{s}}
\;:=\;
\Bigl\{\,w\in\mathbb{R}^{T^{s}}\;:\;w_{t}\ge0,\;\;\sum_{t=1}^{T^{s}}w_{t}=1\,\Bigr\}
\]
such that
\[
\bar r^{\,s}
\;=\;
\sum_{t=1}^{T^{s}} w_{t}\,\hat r_{t}.
\]
Further considerations on the convex-hull assumption are deferred to
Appendix~A.

Given the conditioned sub-sample
\(\{\hat r_{t}\}_{t=1}^{T^{s}}\subset\mathbb{R}^{d}\) and a candidate
vector \(x\) in its convex hull, we wish to determine
\begin{equation}\label{opt2_rec}
    \max_{w\in\Delta_{T^{s}}}\;L(w_{1},\dots,w_{T^{s}})
    \quad\text{subject to}\quad
    \sum_{t=1}^{T^{s}}w_{t}\,\hat r_{t}=x,
\end{equation}
where \(L(w)=\prod_{t=1}^{T^{s}}w_{t}\) is the empirical likelihood
associated with the weight vector \(w\), and \(\Delta_{T^{s}}\) is the
standard simplex in \(\mathbb{R}^{T^{s}}\). The estimator
\(\bar r^{\,s}\) is then defined as the value of \(x\) maximising the
\emph{profile empirical likelihood}
\[
\ell(x)
\;:=\;
\max_{w\in\Delta_{T^{s}}}\;\bigl\{\,L(w)\;:\;{\textstyle\sum_{t}}w_{t}\hat r_{t}=x\,\bigr\},
\]
namely
\begin{equation}\label{double_max}
\bar r^{\,s}
\;\in\;
\arg\max_{x}\;\ell(x).
\end{equation}
The function \(\ell(\cdot)\) assigns to every candidate \(x\) the
largest empirical likelihood attainable by any convex combination of
the observed returns whose centroid is \(x\). It can be read as a
plausibility surface on the convex hull: candidates that can be
represented only by highly unbalanced combinations --- with one or a
few observations carrying most of the weight --- inherit a low
empirical likelihood, whereas candidates that admit nearly uniform
representations are deemed more plausible. In the standard
empirical-likelihood terminology~\cite{owenn}, the nuisance variables
\(w\) are eliminated by an inner maximisation step, leaving a
finite-dimensional surface \(\ell(x)\) on which to perform the outer
search.

\begin{proposition}\label{prop:empmean}
Let \(\{\hat r_{t}\}_{t=1}^{T^{s}}\subset\mathbb{R}^{d}\) be the
conditioned sub-sample, and let \(\bar r^{\,s}\) be defined as
in~\ref{double_max}. Then the optimal weight vector
solving~\ref{opt2_rec} at \(x=\bar r^{\,s}\) is the uniform vector
\[
\bar w_{t}\;=\;\frac{1}{T^{s}},
\qquad t=1,\dots,T^{s},
\]
and the corresponding estimator admits the closed form
\[
\bar r^{\,s}
\;=\;
\frac{1}{T^{s}}\sum_{t=1}^{T^{s}}\hat r_{t}.
\]
Moreover, the maximiser is unique.
\end{proposition}

\begin{proof}
The proof proceeds in three steps.

\medskip
\noindent\textit{Reduction of the two-level problem.}
The estimator \(\bar r^{\,s}\) is, by construction, the solution of a
nested optimisation: an outer maximisation over the candidate \(x\)
and an inner maximisation profiling out the weights compatible with
\(x\). Since the moment constraint
\(\sum_{t}w_{t}\hat r_{t}=x\) is the only coupling between the two
levels, the joint problem
\[
\max_{\substack{x\in\mathbb{R}^{d}\\[1pt] w\in\Delta_{T^{s}}}}
\;L(w)
\qquad\text{s.t.}\qquad
\sum_{t=1}^{T^{s}}w_{t}\,\hat r_{t}=x
\]
is equivalent to~\ref{double_max}, with the constraint now acting as
the \emph{definition} of \(x\) in terms of \(w\) rather than as a
restriction on the feasible set. Eliminating \(x\) by direct
substitution reduces the entire problem to an unconstrained
empirical-likelihood maximisation on the simplex,
\begin{equation}\label{el_inner_p}
\bar w
\;\in\;
\arg\max_{w\in\Delta_{T^{s}}}\;L(w),
\qquad
\bar r^{\,s}
\;=\;
\sum_{t=1}^{T^{s}}\bar w_{t}\,\hat r_{t}.
\end{equation}

\medskip
\noindent\textit{Existence and uniqueness of the maximiser.}
Since the logarithm is strictly increasing, the maximisation of
\(L(w)\) is equivalent to that of \(\log L(w)=\sum_{t}\log w_{t}\).
The function \(\sum_{t}\log w_{t}\) is strictly concave on the
relative interior of \(\Delta_{T^{s}}\) and diverges to \(-\infty\)
on the relative boundary, where at least one weight vanishes. By
strict concavity, any stationary point in the interior is the unique
global maximiser, and no boundary point can be optimal. In
particular, the non-negativity constraints \(w_{t}\ge0\) are
automatically inactive, and the only effective constraint is the
unit-sum condition \(\sum_{t}w_{t}=1\).

\medskip
\noindent\textit{Explicit computation via Lagrange
multipliers.}
Introduce the multiplier \(\lambda\in\mathbb{R}\) associated with the
unit-sum constraint, and define the Lagrangian
\[
\mathcal{L}(w,\lambda)
\;=\;
\sum_{t=1}^{T^{s}}\log w_{t}
\;+\;
\lambda\!\left(1-\sum_{t=1}^{T^{s}}w_{t}\right).
\]
The first-order conditions read
\[
\frac{\partial \mathcal{L}}{\partial w_{t}}
\;=\;
\frac{1}{w_{t}}-\lambda
\;=\;0
\qquad\Longleftrightarrow\qquad
w_{t}\;=\;\frac{1}{\lambda},
\qquad t=1,\dots,T^{s},
\]
and are symmetric in the index \(t\); they therefore imply
\(w_{1}=w_{2}=\cdots=w_{T^{s}}\). Imposing the unit-sum constraint
yields
\[
\sum_{t=1}^{T^{s}}\frac{1}{\lambda}\;=\;1
\qquad\Longleftrightarrow\qquad
\lambda\;=\;T^{s},
\]
so that
\[
\bar w_{t}\;=\;\frac{1}{T^{s}},
\qquad t=1,\dots,T^{s}.
\]
By Step~2 this is the unique maximiser of~\ref{el_inner_p}, and
substitution into the definition of \(\bar r^{\,s}\) gives
\[
\bar r^{\,s}
\;=\;
\sum_{t=1}^{T^{s}}\bar w_{t}\,\hat r_{t}
\;=\;
\frac{1}{T^{s}}\sum_{t=1}^{T^{s}}\hat r_{t},
\]
which is the desired closed form.
\end{proof}

The candidate \(x\) corresponding to the unconstrained optimum is the
empirical mean itself, and no alternative weight vector on the simplex
can yield a higher empirical likelihood. We finally observe that this
result is in line with the maximum-entropy principle: in the absence
of further information about the distribution, the uniform weights
constitute the \emph{least informative} choice compatible with the
constraints imposed by the sample.

\subsection{Parametric Sampling}

Once the estimate \(\bar r^{\,s}\) has been obtained, we proceed by
sampling returns in a neighbourhood of \(\bar r^{\,s}\), favouring a
parametric approach that assigns a specific law to the returns under
stress; a fully nonparametric scheme remains nonetheless available
(cf.\ Section~\ref{nonpar-sampling}). The generated samples define
plausible trajectories of the \(d\) asset classes under stress, while
respecting the historical correlation structure.

Formally, given the historical sample
\[
\{\hat r_{1},\dots,\hat r_{T}\}\subset\mathbb{R}^{d},
\]
the objective is to obtain a prescribed number \(\tau\in\mathbb{N}\) of
realisations in stressed conditions
\(\{\hat r^{\,s}_{1},\dots,\hat r^{\,s}_{\tau}\}\subset\mathbb{R}^{d}\)
such that
\[
\mathbb{E}\bigl[\hat r^{\,s}_{k}\bigr]=\bar r^{\,s},
\qquad
\rho\bigl[\hat r^{\,s}_{k}\bigr]=\rho\bigl[r\bigr],
\qquad k=1,\dots,\tau,
\]
where \(\rho[X]\) denotes the correlation matrix of the random vector
\(X\). The two specifications proposed below (Sections~\ref{gauss}
and~\ref{tstud}) differ in the choice of the parametric law for
\(\hat r^{\,s}_{k}\); both are centred at \(\bar r^{\,s}\) and
calibrated using the local covariance estimate \(\hat\Sigma^{\,s}\)
computed on the conditioned sub-sample.

\subsubsection{Gaussian Sampling}\label{gauss}

We assume that, locally in a neighbourhood of \(\bar r^{\,s}\), the
returns can be approximated by a multivariate normal distribution; the
overall model is therefore semiparametric. Note that Gaussianity is invoked only as a \emph{local} model around the estimate
\(\bar r^{\,s}\), not as a global assumption on the distribution of
returns. More precisely, in such a neighbourhood the density is
represented as
\[
\mathcal{N}_d\!\bigl(\bar r^{\,s},\,\hat\Sigma^{\,s}\bigr),
\]
where \(\hat\Sigma^{\,s}\) is the local estimate of the covariance
matrix, assumed symmetric and positive definite. This assumption makes
it possible to enrich the historical series with plausible --- and
possibly unobserved --- realisations, and is particularly useful in the
presence of short time windows, when the number of factors exceeds the
number of available observations, or whenever the empirical symmetry
of the fluctuations around the stress scenario is to be preserved.

Formally, one generates a sequence of independent and identically
distributed vectors
\[
\hat r^{\,s}_{k}\;\sim\;\mathcal{N}_d\!\bigl(\bar r^{\,s},\,\hat\Sigma^{\,s}\bigr),
\qquad k=1,\dots,\tau,
\]
which yields as output the time series of the \(d\) asset classes
under stress.

\subsubsection{Student-\emph{t} Sampling}\label{tstud}

In order to capture, locally, heavier tails than the Gaussian case
would allow, we approximate the local density of returns by a
multivariate Student-\emph{t} distribution centred at \(\bar r^{\,s}\).
As in the Gaussian case, the assumption is purely local: no global
Student-\emph{t} structure is imposed, but only a representation valid
in a neighbourhood of \(\bar r^{\,s}\). This choice parametrises the
tail thickness through the degrees of freedom \(\nu>2\), while
preserving the dependence structure encoded by the scale matrix
\(\hat\Sigma^{\,s}\).

Sampling from a multivariate Student-\emph{t} distribution can be
implemented through the standard \emph{scale-mixture} representation
combining a Gaussian and a chi-squared distribution. Formally, for
\(k=1,\dots,\tau\) one independently generates
\[
q_{k}\;\sim\;\mathcal{N}_d\bigl(0,\,\hat\Sigma^{\,s}\bigr),
\qquad
p_{k}\;\sim\;\chi^{2}_{\nu},
\]
and sets
\[
\hat r^{\,s}_{k}
\;=\;
\bar r^{\,s}\;+\;q_{k}\,\sqrt{\frac{\nu}{p_{k}}}
\;\sim\;
t^{\,d}_{\nu}\bigl(\bar r^{\,s},\,\hat\Sigma^{\,s}\bigr).
\]
By construction \(\mathbb{E}[\hat r^{\,s}_{k}]=\bar r^{\,s}\) (provided
\(\nu>1\)), whereas the covariance matrix of the samples is
\[
\mathrm{Cov}\bigl[\hat r^{\,s}_{k}\bigr]
\;=\;
\frac{\nu}{\nu-2}\,\hat\Sigma^{\,s},
\qquad \nu>2;
\]
it follows that \(\hat\Sigma^{\,s}\) is to be interpreted as a scale
matrix rather than as the covariance matrix of the samples themselves.
If one wishes the samples to exhibit a covariance exactly equal to a
prescribed target matrix \(\hat\Sigma^{\,s}\), it is sufficient to
rescale the scale matrix as
\(\tfrac{\nu-2}{\nu}\hat\Sigma^{\,s}\). The output is a time series
\((\hat r^{\,s}_{1},\dots,\hat r^{\,s}_{\tau})\in\mathbb{R}^{d\times\tau}\)
that preserves the correlation structure across factors while
displaying more pronounced tails than the Gaussian scheme of
Section~\ref{gauss}.
\section{Nonparametric Approach}\label{nonpar-sampling}
The nonparametric RST procedure admits a fully data-driven
sampling scheme in which no parametric assumption is made at any stage. Having
obtained the estimate \(\bar r^{\,s}\) as in Section~\ref{rst}, the
underlying idea is to draw observations \emph{with replacement} from
the conditioned sub-sample, with selection probabilities concentrated
in a neighbourhood of \(\bar r^{\,s}\).

In order to define such a neighbourhood without committing to any
distributional form, we measure proximity through the \emph{Mahalanobis
distance} associated with the local covariance estimate
\(\hat\Sigma^{\,s}\). The choice is natural in this context:
the Mahalanobis distance standardises distances along the principal
directions of \(\hat\Sigma^{\,s}\) and is invariant under affine
reparametrisations of the return vector, so that the notion of locality
does not depend on the scale of the individual asset classes. Formally,
for any \(x,y\in\mathbb{R}^{d}\) we set
\[
d(x,y)
\;:=\;
\sqrt{(x-y)^{\top}\bigl(\hat\Sigma^{\,s}\bigr)^{-1}(x-y)},
\]
where \(\hat\Sigma^{\,s}\) is assumed symmetric and positive definite.

The selection mechanism is governed by a single hyperparameter, the
\emph{Mahalanobis radius} \(\varepsilon>0\), which controls the degree
of locality of the sampling. Let
\[
B(\bar r^{\,s},\varepsilon)
\;:=\;
\bigl\{\,x\in\mathbb{R}^{d}\;:\;d(x,\bar r^{\,s})\le\varepsilon\,\bigr\}
\]
denote the Mahalanobis ball of radius \(\varepsilon\) centred at the
estimate \(\bar r^{\,s}\). Within the conditioned sub-sample
\(\{\hat r_{t}\}_{t=1}^{T^{s}}\), we retain those observations that
fall inside the ball,
\[
\mathcal{T}^{s}_{\varepsilon}
\;:=\;
\bigl\{\,t\in\{1,\dots,T^{s}\}\;:\;\hat r_{t}\in B(\bar r^{\,s},\varepsilon)\,\bigr\},
\]
which we assume to be non-empty (otherwise the radius \(\varepsilon\)
needs to be enlarged). The radius \(\varepsilon\) governs the
classical bias--variance trade-off: small values concentrate the
sampling on observations very close to \(\bar r^{\,s}\) at the price
of a reduced effective sample size, whereas large values increase the
effective sample size at the price of admitting observations farther
from the conditioned regime.

Within \(\mathcal{T}^{s}_{\varepsilon}\), observations are weighted
inversely with their Mahalanobis distance from \(\bar r^{\,s}\):
denoting \(d_{t}:=d(\hat r_{t},\bar r^{\,s})\), we define the
selection probabilities
\[
\pi_{t}
\;:=\;
\frac{d_{t}^{\,-1}}{\displaystyle\sum_{j\in\mathcal{T}^{s}_{\varepsilon}} d_{j}^{\,-1}},
\qquad
t\in\mathcal{T}^{s}_{\varepsilon},
\]
so that observations closer to \(\bar r^{\,s}\) receive higher
probability mass. The vector
\((\pi_{t})_{t\in\mathcal{T}^{s}_{\varepsilon}}\) defines a discrete
probability distribution on the local sub-sample and can be regarded as
a hyperbolic-kernel reweighting of the empirical distribution, with
bandwidth controlled by \(\varepsilon\).

Finally, we draw \(\tau\) vectors with replacement from
\(\{\hat r_{t}\}_{t\in\mathcal{T}^{s}_{\varepsilon}}\) according to the
distribution \((\pi_{t})\):
\[
\hat r^{\,s}_{k}
\;\stackrel{\text{i.i.d.}}{\sim}\;
\sum_{t\in\mathcal{T}^{s}_{\varepsilon}}\pi_{t}\,\delta_{\hat r_{t}},
\qquad k=1,\dots,\tau,
\]
where \(\delta_{x}\) denotes the Dirac mass at \(x\). The output is the
matrix
\((\hat r^{\,s}_{1},\dots,\hat r^{\,s}_{\tau})\in\mathbb{R}^{d\times\tau}\),
whose columns constitute the simulated time series of the \(d\) asset
classes under stress.

Unlike the parametric schemes of Sections~\ref{gauss} and~\ref{tstud},
the present procedure involves no covariance estimate in the generation
of the samples themselves: the simulated vectors are drawn directly
from historical observations and therefore inherit, by construction,
the joint dependence structure of the local sub-sample. The expected
value of \(\hat r^{\,s}_{k}\) under the sampling distribution equals
the weighted mean
\(\sum_{t\in\mathcal{T}^{s}_{\varepsilon}}\pi_{t}\,\hat r_{t}\),
which approximates --- though does not generally coincide with --- the
estimate \(\bar r^{\,s}\); the approximation tightens as
\(\varepsilon\) decreases and the local sub-sample concentrates around
the conditioned mean.
\section{Results}\label{ch:applications}

In this section we present several application contexts for the
procedures developed in the previous sections. Section~\ref{test}
proposes a generic illustrative experiment, designed to familiarise the
reader with the operational aspects of the method, while
Section~\ref{test_comp} compares the alternative modelling choices
introduced in Sections~\ref{para} and~\ref{rst}.

The convention adopted throughout this work is to express prices in
terms of relative variations and interest rates in terms of absolute
first differences:
\[
\Delta i_{t} = i_{t}-i_{t-1},
\qquad
\Delta p_{t} = \frac{p_{t}}{p_{t-1}}-1,
\]
where \(i_{t}\) denotes a rate and \(p_{t}\) a price. 

The semiparametric and nonparametric variants of the procedure may
overfit the idiosyncratic noise contained in historical stress
episodes. Such events --- typically triggered by news flow or shifts
in market sentiment --- generate isolated and nearly independent
shocks, which can distort the estimation of aggregate market
dynamics, and in particular of the cross-asset dependence structure.
To mitigate this phenomenon one may resort to a time-block
aggregation, in which the conditioning level \(s\) refers to a
longer horizon rather than to a single trading day. For instance, on
a weekly basis (five trading days), a cumulative shock \(S\) prescribed
over a window of \(T_{\text{win}}\) trading days can be expressed in
weekly-equivalent terms as
\[
s = (1+S)^{5/T_{\text{win}}}-1.
\]
The aggregation horizon must be kept consistent across the entire
analysis: all returns being compared --- historical and simulated ---
must be referred to the same temporal base, otherwise the conditioning
step loses interpretability.

\subsection{Validation of the Procedure}\label{test}

We now propose an experiment in support of the RST procedure, applied
to a set of \(1000\) simulations on five time series
spanning five years of market data (\(T=5\times 252\) trading days).
The selected time series correspond to indices of representative asset
classes, as reported in Table~\ref{tab:ticker}.

\begin{table}[H]
\centering
\begin{tabular}{ll}
\hline
\textbf{Time Series}      & \textbf{Ticker}  \\
\hline
European equity                          & M7EU Index    \\
US equity                                & GDDLNA Index  \\
European government fixed income         & EG00 Index    \\
Emerging-market government fixed income  & EMGB Index    \\
1Y Italian government yield             & GBOTG12M Index\\
\hline
\end{tabular}
\caption{Indices selected for the validation of the methodology.}
\label{tab:ticker}
\end{table}

The results that follow support the capability of the model to
estimate, jointly and coherently, the impact of a stress event on
multiple asset classes, producing a unified scenario for both returns
and interest rates.

The first scenario consists of a cumulative shock of \(-11\%\) on the
European government bond index, distributed over a window of \(58\)
trading days. The magnitude and the horizon of the shock are
calibrated to mirror the conditions observed between August and
October~2022. The fluctuation is, by construction, absorbed exactly by
the shocked series, while the remaining asset classes respond
endogenously through the estimated dependence structure, as summarised
by the empirical quantiles in Table~\ref{tab:shock_performance}.

\begin{table}[H]
  \centering
  \begin{tabular}{llrrrrr}
    \toprule
    Asset    & Scenario & Q5\%  & Q25\% & Q50\% & Q75\% & Q95\% \\
    \midrule
    Gov EU   & Baseline &  -3\% &   0\% &   1\% &   2\% &   4\% \\
             & Shocked  & -19\% & -14\% & -11\% &  -8\% &  -3\% \\
    \midrule
    Gov EM   & Baseline &  -6\% &   0\% &   2\% &   3\% &   6\% \\
             & Shocked  & -12\% &  -4\% &   0\% &   5\% &  11\% \\
    \midrule
    Eq EU    & Baseline & -15\% &  -3\% &   2\% &   6\% &  13\% \\
             & Shocked  & -22\% &  -2\% &  12\% &  26\% &  49\% \\
    \midrule
    Eq USA   & Baseline & -13\% &  -3\% &   3\% &   7\% &  14\% \\
             & Shocked  & -34\% &  -7\% &  11\% &  30\% &  62\% \\
    \bottomrule
  \end{tabular}
  \caption{Empirical quantiles of the cumulative percentage variations
  over the period subject to a shock of \(-11\%\) on the European
  government bond index, under the baseline and the stressed scenario.}
  \label{tab:shock_performance}
\end{table}

Two features of Table~\ref{tab:shock_performance} deserve emphasis.
First, the median response of the equity indices is markedly upward,
consistently with the historical correlation pattern between European
fixed-income and equity returns observed in the calibration sample.
Second, the \(5\%\) quantile under the stressed scenario remains more
negative than the corresponding baseline quantile, and the
inter-quantile range expands appreciably; the model thus captures the
standard risk--reward asymmetry of stress regimes, in which a stronger
central tendency is accompanied by a generalised increase in the
dispersion of outcomes and, consequently, in tail risk.

The emerging-market bond index displays a more contained downward
correction, consistent with its partial decoupling from European
fixed-income dynamics. Its conditional distribution, however, exhibits
a noticeably larger variability than under the baseline scenario,
coherently with the expected risk-on/risk-off behaviour in stressed
market regimes.

\begin{figure}[H]
  \centering
  \includegraphics[width=1\textwidth]{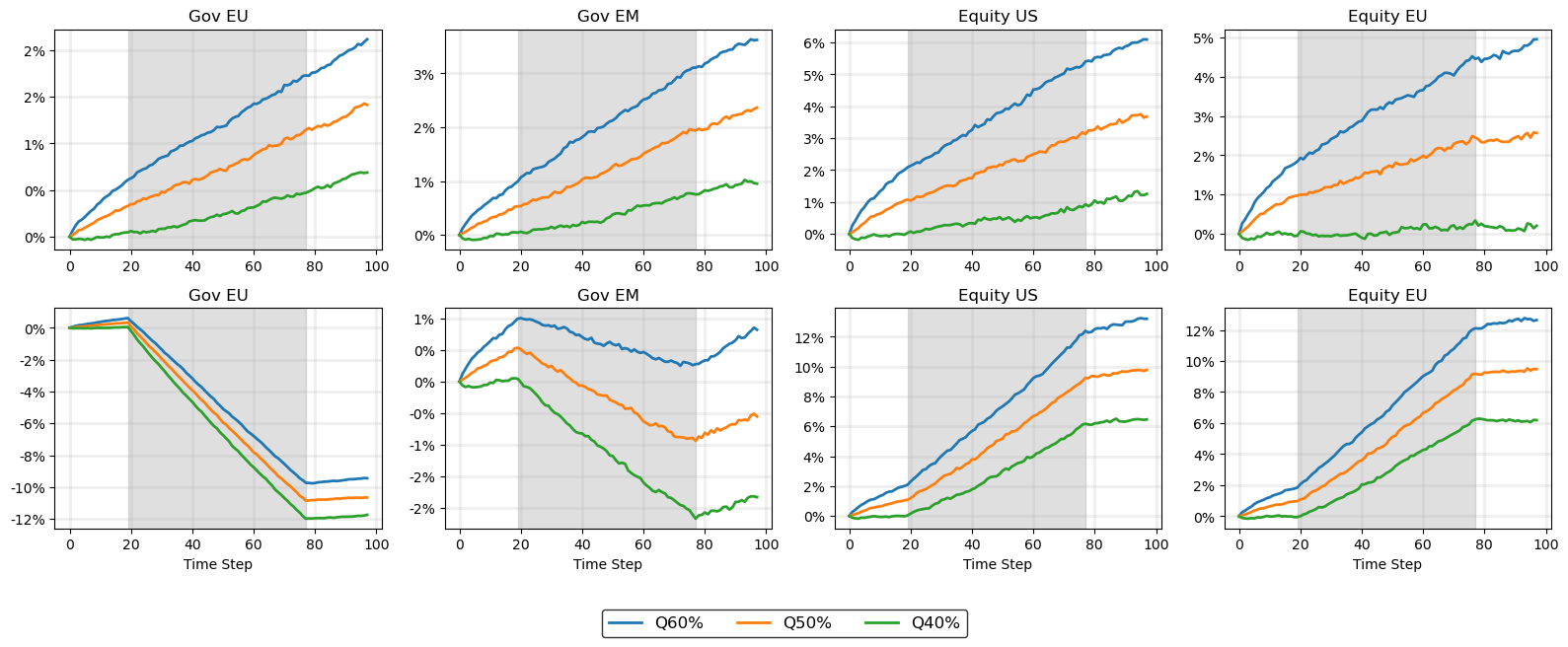}
  \caption{Baseline scenario (top row) and stress scenario (bottom row)
  in the case of a \(-11\%\) shock on the European government bond
  index. The shaded band identifies the time interval subject to the
  shock.}
  \label{fig:rst}
\end{figure}

Figure~\ref{fig:rst} reports the graphical output of the procedure.
The first column corresponds to the asset class subject to the
imposed shock and is characterised by visibly reduced variability ---
as expected, since its trajectory is constrained by construction to
the target level. The shaded band in each panel marks the stress
window: outside it, the simulated trajectories coincide with the
historical ones, whereas inside it the procedure replaces the original
returns with realisations coherent with the imposed shock.

\subsection{Comparison Across Approaches}\label{test_comp}

We now compare the parametric, semiparametric and nonparametric
variants of the procedure on a common experimental setup. For the sake
of legibility, the analysis is restricted to three asset classes: the
European government bond index, the European equity index and the
one-year Italian government yield. A cumulative shock of \(-30\%\) is
imposed on the equity index, distributed over a window of \(63\)
trading days (approximately three months).

\paragraph{Parametric approach.}
Figure~\ref{fig:para} reports the simulated trajectories obtained
under the fully Gaussian parametric variant of Section~\ref{para}. The
sustained drop in the equity index is accompanied by a positive
response of the government bond index, which records a median
cumulative total return of \(+2\%\) at the end of the stress window,
while the one-year Italian government yield remains substantially
unchanged.

\begin{figure}[H]
  \centering
  \includegraphics[width=1\textwidth]{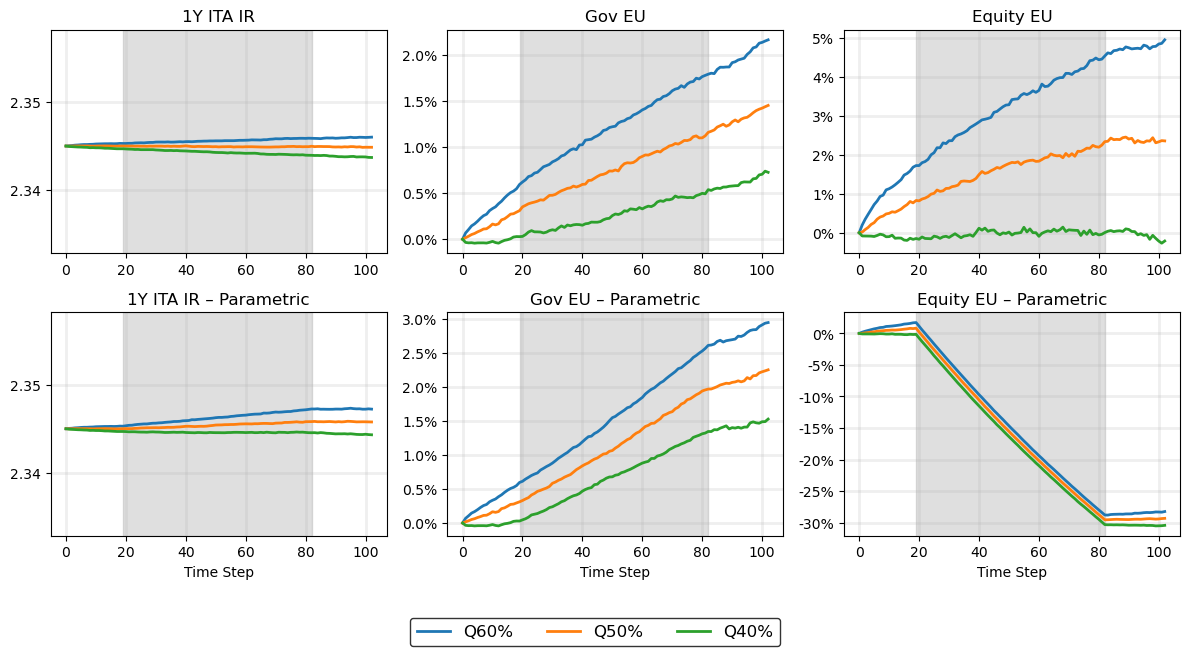}
  \caption{Baseline scenario (top row) and stress scenario (bottom row)
  generated by the parametric RST method, under a \(-30\%\) shock on
  the European equity index. The shaded band identifies the time
  interval subject to the shock.}
  \label{fig:para}
\end{figure}

\paragraph{Semiparametric.}
The semiparametric variants remain close to the Gaussian benchmark.
Figures~\ref{fig:semipara1} and~\ref{fig:semipara2} report, respectively,
the Gaussian and the Student-\emph{t} semiparametric schemes. The fully parametric
estimator delivers trajectories that are visibly less affected by
sampling noise, a direct consequence of the closed-form expression of
the underlying density; the dispersion induced by the shock is
correspondingly lower than in the semiparametric counterparts. The
Student-\emph{t} variant produces a fan of trajectories whose central
tendency is comparable to the Gaussian one, while individual paths
exhibit heavier tails --- a direct consequence of the parametric
flexibility introduced by the degrees of freedom \(\nu\).

\begin{figure}[H]
  \centering
  \includegraphics[width=1\textwidth]{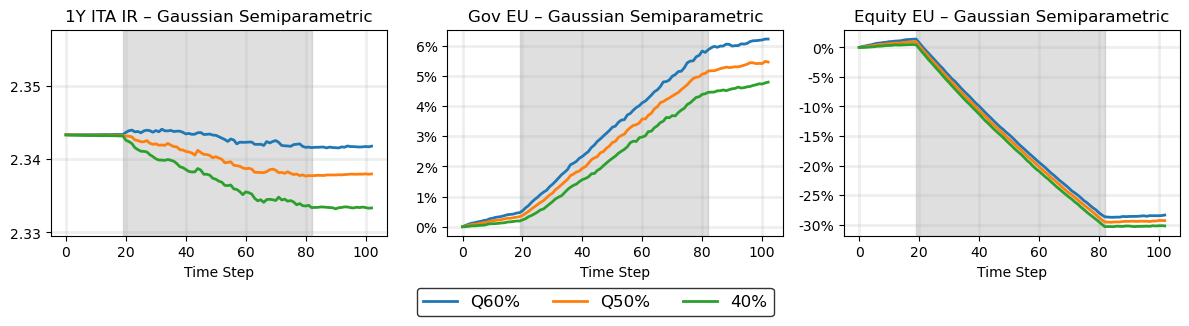}
  \caption{Stress scenario generated by the semiparametric Gaussian
  RST method, under a \(-30\%\) shock on the European equity index.
  The shaded band identifies the time interval subject to the shock.}
  \label{fig:semipara1}
\end{figure}

\begin{figure}[H]
  \centering
  \includegraphics[width=1\textwidth]{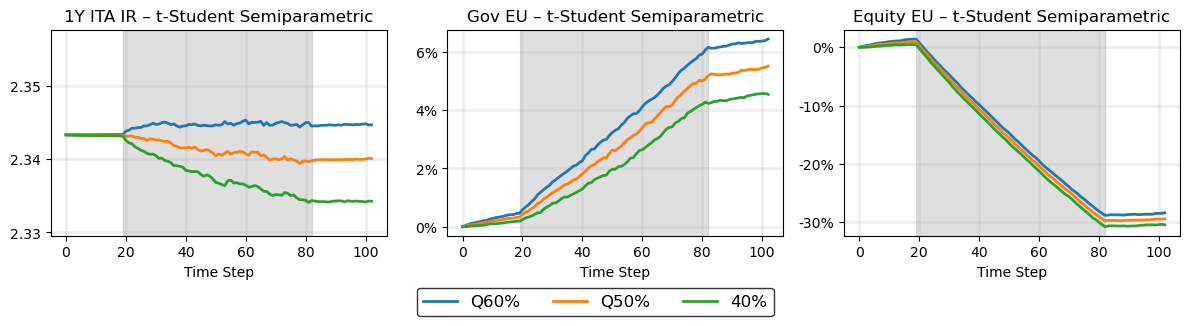}
  \caption{Stress scenario generated by the semiparametric RST method
  with Student-\emph{t} sampling, under a \(-30\%\) shock on the
  European equity index. The shaded band identifies the time interval
  subject to the shock.}
  \label{fig:semipara2}
\end{figure}

\paragraph{Nonparametric approach.}
Figure~\ref{fig:nonpa} reports the nonparametric variant, whose
output differs qualitatively from the previous schemes. The most
striking feature is a marked reduction in the simulated dispersion ---
a direct consequence of the sampling mechanism: by drawing with
replacement from a Mahalanobis neighbourhood of \(\bar r^{\,s}\), the
procedure restricts itself to a comparatively small set of historical
observations. Since events close to \(\bar r^{\,s}\) are rare by
definition, the nonparametric scheme may sample repeatedly from the
same observations, producing trajectories with reduced intra-shock
variability. This effect represents the operational counterpart of the
bias--variance trade-off discussed in Section~\ref{nonpar-sampling}:
the absence of any distributional assumption protects against
parametric misspecification at the cost of a reduced effective sample
size in the conditioned region.

\begin{figure}[H]
  \centering
  \includegraphics[width=1\textwidth]{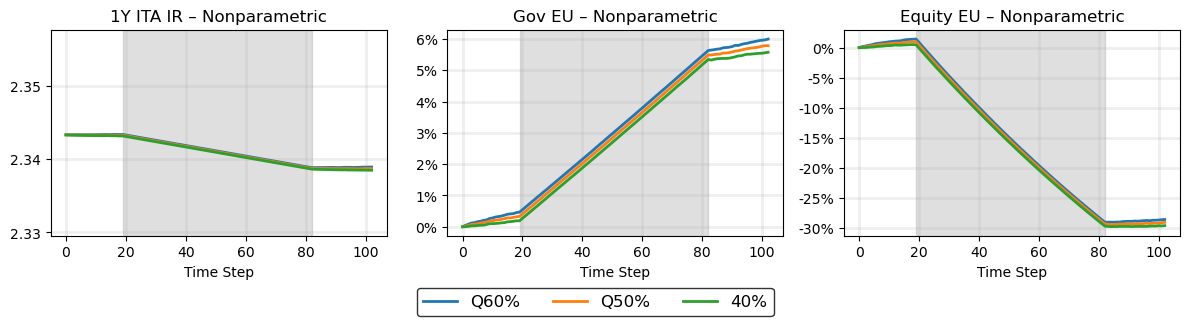}
  \caption{Stress scenario generated by the nonparametric RST method,
  under a \(-30\%\) shock on the European equity index. The shaded
  band identifies the time interval subject to the shock.}
  \label{fig:nonpa}
\end{figure}
\section{Conclusions}
We have developed an RST framework enriched by a dedicated sampling
phase, which makes it possible to generate plausible market-wide
stress scenarios from a single exogenous perturbation imposed on one
asset class. The methodology has been formulated in three variants ---
parametric, semiparametric and nonparametric --- so as to
accommodate different operational requirements while preserving a
common conditioning logic based on the empirical-likelihood reading
of the optimal scenario.

The proposed models have been tested on real market data, using
representative indices of distinct asset classes. The experimental
results support the effectiveness of the framework, both in terms of
the expected magnitude of the simulated responses and in terms of the
economic coherence of the cross-asset dynamics induced by the imposed
shock. In particular, the joint behaviour of fixed-income, equity and
interest-rate factors observed across the simulated trajectories is
consistent with stylised market regularities and with the historical
co-movement structure encoded in the calibration sample.

The empirical analysis also highlights a limitation of the fully
nonparametric variant, whose accuracy proves to be more sensitive
both to the size of the available sample and to potential asymmetries
in the empirical distribution of returns. As discussed in
Section~\ref{nonpar-sampling}, this is a structural feature of any
inverse-distance resampling scheme operating in a sparsely populated
region, and not a defect of the specific implementation. A systematic
investigation of regularisation techniques for the nonparametric
scheme --- alternative kernels, adaptive bandwidths, smoothed bootstrap
variants --- falls outside the scope of the present work; the
discussion here is therefore confined to documenting the qualitative
robustness of the semiparametric variant, which emerges as the most
balanced trade-off between distributional flexibility and stability of
the estimated stress scenarios. The parametric variant, in turn,
remains advantageous in those operational contexts in which future
dynamics can reasonably be assumed to be consistent with past
performance, since it delivers smoother trajectories and a sharper
control of the conditional variability.

Compared with the methodology proposed by the Office of Financial
Research, the approach developed in this paper privileges the
operational and practical aspects of stress-scenario generation. In
particular, the methods described in this work lend themselves
naturally to a post-processing role with respect to standard
simulation techniques such as the bootstrap: starting from a base
simulation of the historical series, the RST procedure replaces the
returns inside the stress window with realisations coherent with the
imposed shock and the estimated dependence structure, leaving the
remaining part of the trajectory unaltered. The resulting pipeline is
modular --- the simulation engine and the stress engine remain
decoupled --- and computationally efficient, since the conditioning
step exploits closed-form expressions wherever a Gaussian or
semiparametric structure is adopted.

\appendix
\section{Appendix: Theoretical Complements}\label{appendice:A}

This appendix collects two technical results that complement the
empirical-likelihood approach developed in Section~\ref{rst}.

\subsection{Bootstrap as a Special Case of the Convex-Hull Formulation}
\label{sec:bootstrap_subset}

The empirical-likelihood framework of Section~\ref{rst} searches for the
modal scenario \(\bar r^{\,s}\) within the convex hull of the
conditioned sub-sample, rather than restricting attention to the
finite set of observed returns. The following result formalises the
fact that the latter approach is strictly more conservative.

\begin{proposition}\label{prop:bootstrap_subset}
Let \(\{\hat r_{t}\}_{t=1}^{T^{s}}\subset\mathbb{R}^{d}\) denote the
conditioned sub-sample, and define
\[
U^{(\mathrm{bootstrap})}
\;:=\;
\bigl\{\,\hat r_{t}\;:\;t\in\{1,\dots,T^{s}\}\,\bigr\},
\qquad
U
\;=\;
\Bigl\{\,\sum_{t=1}^{T^{s}}w_{t}\,\hat r_{t}\;:\;w\in\Delta_{T^{s}}\,\Bigr\}.
\]
Then \(U^{(\mathrm{bootstrap})}\subset U\).
\end{proposition}

\begin{proof}
Let \(e^{(t)}\in\mathbb{R}^{T^{s}}\) denote the \(t\)-th vector of the
canonical basis, namely \(e^{(t)}_{j}=1\) if \(j=t\) and
\(e^{(t)}_{j}=0\) otherwise. By construction
\(\sum_{j}e^{(t)}_{j}=1\) and \(e^{(t)}_{j}\ge0\), so that
\(e^{(t)}\in\Delta_{T^{s}}\). The corresponding convex combination
satisfies
\[
\sum_{j=1}^{T^{s}}e^{(t)}_{j}\,\hat r_{j}\;=\;\hat r_{t},
\]
which proves \(\hat r_{t}\in U\) for every \(t\). Strict inclusion
follows by observing that the empirical mean
\(\frac{1}{T^{s}}\sum_{t}\hat r_{t}\) lies in \(U\) by convexity but
coincides with no individual observation \(\hat r_{t}\) unless the
sub-sample is degenerate.
\end{proof}

\begin{remark}
Proposition~\ref{prop:bootstrap_subset} provides the theoretical
ground for the convex-hull assumption used throughout
Section~\ref{rst}: the empirical-likelihood approach inherits the
consistency of the bootstrap --- the historical observations remain
admissible scenarios --- while admitting a strictly richer set of
candidate scenarios, namely all the convex combinations of the
observations falling in the conditioned region.
\end{remark}

\subsection{The Empirical-Likelihood Substitution}\label{sec:el_substitution}

Problem~\ref{opt} prescribes the maximisation of the unknown joint
density \(f_{r}\) over the convex hull of the conditioned sub-sample.
Since \(f_{r}\) is not available in closed form within the
nonparametric setting of Section~\ref{rst}, direct evaluation of
problem~\ref{opt} is infeasible.

The empirical-likelihood methodology of Owen~\cite{owenn} replaces
\(f_{r}\) by the empirical likelihood
\(L(w)=\prod_{t=1}^{T^{s}}w_{t}\) associated with the discrete
distribution \(P_{w}=\sum_{t}w_{t}\,\delta_{\hat r_{t}}\) supported on
the observed sub-sample. The corresponding \emph{profile empirical
likelihood} is defined as
\[
\ell(x)
\;:=\;
\max_{w\in\Delta_{T^{s}}}
\Bigl\{\,
\prod_{t=1}^{T^{s}}w_{t}
\;:\;
\sum_{t=1}^{T^{s}}w_{t}\,\hat r_{t}=x
\,\Bigr\},
\qquad x\in U.
\]
The substitution of \(f_{r}\) by \(\ell\) transforms problem~\ref{opt}
into problem~\ref{opt2_rec}. The following proposition establishes the
structural properties of \(\ell\) that justify this substitution.

\begin{proposition}[Structural properties of the profile EL]\label{prop:profile_EL}
The profile empirical likelihood \(\ell:U\to[0,+\infty)\) satisfies:
\begin{enumerate}
\item[(i)] for every \(x\) in the relative interior of \(U\) the
maximum defining \(\ell(x)\) is attained at a unique interior point of
\(\Delta_{T^{s}}\); in particular \(\ell(x)>0\);
\item[(ii)] \(\log\ell\) is strictly concave on the relative interior
of \(U\);
\item[(iii)] \(\ell\) admits a unique global maximiser, which coincides
with the empirical mean
\(\bar r=\frac{1}{T^{s}}\sum_{t}\hat r_{t}\)
and is attained at the uniform weight vector \(w_{t}=1/T^{s}\).
\end{enumerate}
\end{proposition}

By Proposition~\ref{prop:profile_EL}, the profile empirical likelihood
\(\ell\) inherits the qualitative features expected of a probability
density supported on \(U\): non-negativity on \(U\), vanishing on
the relative boundary, strict log-concavity on the interior, and a
unique mode. These properties make \(\ell\) a structurally consistent
proxy for \(f_{r}\) in the optimisation problem, and reduce the
original maximisation in problem~\ref{opt} to a finite-dimensional, strictly
concave problem with linear constraints (problem~\ref{opt2_rec}).

\section{Appendix: Local Covariance Estimation}\label{appendice:B}

The semiparametric and nonparametric procedures developed in
Sections~\ref{rst} and~\ref{nonpar-sampling} rely on a covariance
estimate that reflects the local market regime in a neighbourhood of
the conditioned scenario \(\bar r^{\,s}\), rather than the
unconditional historical dependence. The distinction is substantive
rather than incidental. Cross-asset correlations are widely documented
to be regime-dependent, and in particular to intensify during periods
of distress --- a phenomenon commonly referred to as
\emph{correlation breakdown}. A global covariance estimator computed
over the entire historical sample averages across regimes and is
therefore inadequate as an input to a stress-scenario generator, whose
purpose is precisely to characterise the dependence in the conditioned
region. This appendix collects the operational details of the local
estimator \(\hat\Sigma^{\,s}\) employed throughout the paper, together
with a Bayesian regularisation scheme based on the conjugate
inverse-Wishart prior, intended to mitigate the small-sample
instabilities that typically arise when restricting estimation to a
conditioned sub-sample.

\subsection{Mahalanobis-Based Estimation}

Let \(\{\hat r_{t}\}_{t=1}^{T}\subset\mathbb{R}^{d}\) denote the
available historical sample of daily returns, and let
\(\bar r^{\,s}\in\mathbb{R}^{d}\) denote the modal scenario --- as
obtained either from the closed-form solution of Section~\ref{para},
or from the empirical-likelihood procedure of Section~\ref{rst}. Let
\[
\hat\Sigma^{\mathrm{glob}}
\;:=\;
\frac{1}{T-1}\sum_{t=1}^{T}(\hat r_{t}-\bar r)(\hat r_{t}-\bar r)^{\top},
\qquad
\bar r \;:=\; \frac{1}{T}\sum_{t=1}^{T}\hat r_{t},
\]
denote the unconditional sample covariance, assumed symmetric and
positive definite.

The locality of the estimation is governed by the squared Mahalanobis
distance of each historical observation from the modal scenario,
computed with respect to the global covariance:
\[
d_{t}^{\,2}
\;:=\;
(\hat r_{t}-\bar r^{\,s})^{\top}\bigl(\hat\Sigma^{\mathrm{glob}}\bigr)^{-1}(\hat r_{t}-\bar r^{\,s}),
\qquad t=1,\dots,T.
\]
We retain those observations whose squared Mahalanobis distance falls
below the empirical \(5\%\) quantile of the distribution of
\(\{d_{t}^{\,2}\}\):
\[
q_{5\%}
\;:=\;
Q_{5\%}\bigl(\{d_{1}^{\,2},\dots,d_{T}^{\,2}\}\bigr),
\qquad
\mathcal{T}^{\,s}_{\mathrm{loc}}
\;:=\;
\bigl\{\,t\in\{1,\dots,T\}\;:\;d_{t}^{\,2}<q_{5\%}\,\bigr\}.
\]
The threshold level \(5\%\) is a hyperparameter that governs the
standard bias--variance trade-off: lower percentiles concentrate the
estimation on a tighter neighbourhood of \(\bar r^{\,s}\) at the cost
of a smaller effective sample size, while higher percentiles admit a
larger sub-sample at the price of including observations less
representative of the stressed regime.

The local sample covariance is defined as the standard empirical
covariance restricted to the selected sub-sample:
\[
\hat\Sigma^{\,s}_{\mathrm{loc}}
\;:=\;
\frac{1}{|\mathcal{T}^{\,s}_{\mathrm{loc}}|-1}
\sum_{t\in\mathcal{T}^{\,s}_{\mathrm{loc}}}
(\hat r_{t}-\bar r_{\mathrm{loc}})(\hat r_{t}-\bar r_{\mathrm{loc}})^{\top},
\]
where
\[
\bar r_{\mathrm{loc}}
\;:=\;
\frac{1}{|\mathcal{T}^{\,s}_{\mathrm{loc}}|}
\sum_{t\in\mathcal{T}^{\,s}_{\mathrm{loc}}}\hat r_{t}
\]
denotes the empirical mean of the local sub-sample.

In the applications of Section~\ref{para}, the same Mahalanobis-based
selection procedure can be applied to estimate only specific blocks of
the global covariance --- typically the variance of the shocked
component \(\Sigma_{11}^{\,s}\) and the cross-covariance blocks
\(\Sigma_{1,-1}^{\,s},\,\Sigma_{-1,1}^{\,s}\) --- while retaining a
global estimate for the covariance \(\Sigma_{-1,-1}\) of the
non-shocked factors. This hybrid specification preserves the
stability of the marginal estimates on the larger block while letting
the dependence between the shocked and the non-shocked components
reflect the local regime.

\subsection{Bayesian Shrinkage Regularisation}

When the conditioned sub-sample is small --- a typical occurrence in
high-dimensional or rare-event settings --- the local sample
covariance \(\hat\Sigma^{\,s}_{\mathrm{loc}}\) may be ill-conditioned or
even singular. The standard remedy is to shrink the estimate towards
a regularising target, and we follow here a Bayesian formulation based
on the conjugate inverse-Wishart prior for the covariance matrix,
which makes the role of the regularisation hyperparameter transparent.

We adopt the prior
\[
\Sigma^{\mathrm{prior}}
\;:=\;
\frac{\sigma_{s}^{2}}{252}\,I_{d},
\]
where \(I_{d}\) denotes the \(d\)-dimensional identity matrix and
\(\sigma_{s}^{2}\) is a user-supplied scalar interpretable as a target
annualised variance under stress. The factor \(1/252\) rescales the
annualised variance to a daily horizon coherently with the time scale
of the returns. The prior is uninformative \emph{across} asset classes
(its structure is invariant under rotations of the return vector) but
informative \emph{about} the scale of the volatility, and constitutes
the natural channel through which expert judgement on the stress
regime can enter the estimation.

Under the conjugate model
\[
\Sigma
\;\sim\;
\mathcal{IW}_{d}\bigl(\nu_{0}\,\Sigma^{\mathrm{prior}},\,\nu_{0}+d+1\bigr),
\qquad
\hat r_{t}\mid\Sigma
\;\stackrel{\text{i.i.d.}}{\sim}\;\mathcal{N}_{d}(\bar r^{\,s},\Sigma)
\;\;\text{for } t\in\mathcal{T}^{\,s}_{\mathrm{loc}},
\]
where the prior is parametrised so that its mean equals
\(\Sigma^{\mathrm{prior}}\), the posterior distribution is again
inverse-Wishart and admits the closed-form posterior mean
\begin{equation}\label{eq:bayes_shrinkage}
\hat\Sigma^{\,s}
\;:=\;
\frac{\nu_{0}\,\Sigma^{\mathrm{prior}}
      \;+\;n_{\mathrm{loc}}\,\hat\Sigma^{\,s}_{\mathrm{loc}}}
     {\nu_{0}+n_{\mathrm{loc}}},
\qquad
n_{\mathrm{loc}}\;:=\;|\mathcal{T}^{\,s}_{\mathrm{loc}}|.
\end{equation}
The estimator~\ref{eq:bayes_shrinkage} is a convex combination of the
prior scale and of the local sample covariance, with weights given by
the prior pseudo-sample size \(\nu_{0}\) and the effective size of the
local sub-sample \(n_{\mathrm{loc}}\).

The hyperparameter \(\nu_{0}\) admits a direct frequentist
interpretation as the number of \emph{pseudo-observations} carried by
the prior. As \(\nu_{0}\to 0\) the estimator collapses to the local
sample covariance \(\hat\Sigma^{\,s}_{\mathrm{loc}}\), recovering the
fully empirical specification; as \(\nu_{0}\to\infty\) the estimator
tends to the prior \(\Sigma^{\mathrm{prior}}\), recovering the fully
prior-driven specification. Intermediate values interpolate between
the two regimes. The choice of \(\nu_{0}\) is left to the user, and
can be calibrated either on a hold-out sample or fixed at a
conventional value (e.g.\ \(\nu_{0}=d\) or \(\nu_{0}=n_{\mathrm{loc}}\),
yielding equal weights between prior and data).

\newpage
\nocite{*}
\printbibliography[heading=bibintoc]

\end{document}